%% file: main.tex
\documentclass[a4paper]{article}

\usepackage[english]{babel}
\usepackage[utf8x]{inputenc}
\usepackage[T1]{fontenc}
\usepackage[title]{appendix}
\usepackage{framed}
\usepackage{algorithm}
\usepackage{algpseudocode}
\usepackage{color}

\usepackage{float}

 {\endMakeFramed}
\definecolor{mygray}{gray}{0.9}

\usepackage[a4paper,top=3cm,bottom=2cm,left=3cm,right=3cm,marginparwidth=1.75cm]{geometry}

\usepackage{amsmath}
\usepackage{amssymb}
\usepackage{amsthm}
\usepackage{proof}
\usepackage{graphicx}
\usepackage{subcaption}
\usepackage[colorinlistoftodos]{todonotes}
\usepackage[colorlinks=true, allcolors=blue]{hyperref}
\usepackage{tikz-cd}
\usepackage{wrapfig,lipsum}
\usepackage{authblk}

\newcounter{Counter}
\newtheorem{Theorem}[Counter]{Theorem}

\newtheorem{Lemma}[Counter]{Lemma}
\newtheorem{Corollary}[Counter]{Corollary}
\newtheorem{Observation}[Counter]{Observation}

\newcommand{\toffoli}{\mathrm{TOFFOLI}}
\newcommand{\cnot}{\mathrm{CNOT}}
\newcommand{\nnot}{\mathrm{NOT}}

\title{Explicit lower bounds on strong simulation of quantum circuits in terms of $T$-gate count}
\author[1,2]{ Cupjin Huang}
\author[1,3]{Michael Newman}
\author[1]{Mario Szegedy}

\affil[1]{\small \textit{Aliyun Quantum Laboratory, Alibaba Group, Bellevue, WA 98004, USA}}
\affil[2]{\small \textit{Department of Computer Science, University of Michigan, Ann Arbor, MI 48109, USA}}
\affil[3]{\small \textit{Department of Physics, Duke University, Durham, NC, 27708, USA}}
\date{}
\begin{document}
\maketitle
%## Abstract
\input{abstract}

%## Motivation
\input{results}

\input{proof}
\newpage
\input{appendix}
\bibliographystyle{alpha}
\bibliography{biblio}

\end{document}

%% file: abstract.tex
\begin{abstract}
We investigate Clifford+$T$ quantum circuits with a small number of $T$-gates.  Using the sparsification lemma, we identify time complexity lower bounds in terms of $T$-gate count below which a strong simulator would improve on state-of-the-art $3$-SAT solving.
\end{abstract}

%% file: results.tex
\section{Results}

Recent results \cite{bravyi2016improved,  bravyi2018simulation} have shown that a quantum circuit can be strongly simulated in time $O^*(2^{0.47N})$, where 
 $N$ is the number of $T = Z^{1/4}$-gates in an otherwise Clifford circuit\footnote{With slight abuse of notation, we also allow the inverse of the $T$-gate ($T^\dag=P^\dag T$) in the gate set and define the $T$-count of the circuit to be the number of $T$ and $T^\dag$ gates altogether.}. As Clifford + $T$ gates form a universal gate set, this simulation method yields a substantial speed-up on circuits that are predominantly Clifford. We show the following~\footnote{We note that results similar to Theorem 1 have been obtained in~\cite{morimae2019fine} independently.}.

\begin{Theorem}\label{theorem:main}
Assuming the Exponential Time Hypothesis (ETH), there is an $\epsilon>0$ such that any strong simulation that can determine if $\langle 0|C|0|\rangle\neq 0$ of a polynomial-sized quantum circuit $C$ formed from the Clifford+$T$ gate set
with $N$ $T$-gates takes time at least $2^{\epsilon N}$. 
\end{Theorem}

We review the ETH in the next section, but for explicit constants, we have the following theorem.

\begin{Theorem} \label{thm:main1}
Assume that there exists a strong simulator that, for any Clifford+$T$ circuit with $N$ $T$-gates, can determine if $\langle 0|C|0\rangle\neq 0$ in time $O(2^{2.2451\times 10^{-8}N})$. Then this would improve on the current state-of-the-art {\rm 3-SAT} solver by achieving an $O(1.3^{n})$ runtime for $m={\rm poly}(n)$, where $n$ denotes the number of variables of the {\rm 3-SAT} instance and $m$ denotes the number of clauses.
\end{Theorem}

Theorem \ref{theorem:main} relies on the following.

\begin{Lemma}[Corollary 2, \cite{impagliazzo2001problems}] \label{ethcons}
    \label{lem:eth}% Exponential Time Hypothesis w.r.t. length
    Assuming the ETH, there exists constant $a>0$ such that any classical algorithm solving {\rm 3-SAT} instances with length $L$ takes $2^{aL}$ time, where again $L$ is the \emph{length} of the formula.
\end{Lemma}

To compute explicit constants, Theorem \ref{thm:main1} relies on the following.

\begin{Lemma}\label{lem:main}
    Assume that a classical algorithm solves {\rm 3-SAT} in time $O(2^{3.1432\times 10^{-7}L})$, where $L$ is the length of the formula,  $m_{2}$ is the number of $2$-clauses, $m_{3}$ is the number of $3$-clauses, so that $m = m_{2} + m_{3}$ and $L = 2m_{2}+3m_{3}-1$. 
    
    Then one can create a {\rm 3-SAT} solver that achieves an $O(1.3^{n})$ running-time 
    for $m={\rm poly}(n)$, where $n$ denotes the number of variables of the {\rm 3-SAT} instance and $m$ denotes the number of clauses.
\end{Lemma}

In the next section we describe the Sparsification Lemma, which was developed in \cite{impagliazzo2001problems} to 
address the type of reduction found in Lemmas  \ref{ethcons} and  \ref{lem:main}.  The necessity of the Sparsification Lemma comes from reducing the 3-SAT problem to quantum circuits, as the number of $T$ gates in the reduced instance will depend on $L$, the {\em length} of the 3SAT instance, 
and not on $n$, the {\em number of variables}. The reduction from 3-SAT to strong simulation of quantum circuits is described in Section \ref{SATtoT}, while
a full proof of the Sparsification Lemma with explicit constants will be given in the Appendix.

%% file: proof.tex
\section{ETH and the Sparsification Lemma}

The input for the 3-SAT problem with parameters $n$ and $m$ is a Boolean formula
\[
\phi(x_{1},\ldots,x_{n}) = \bigwedge_{i=1}^{m}\left( \bigvee_{j=1}^{k_{i}} l_{i,j}\right) \;\;\;\;\;\;\;\; l_{i,j}\in \{x_{1},\ldots,x_{n}, \neg x_{1},\ldots,\neg x_{n}\}
\]
where 
$l_{i,j}$ are called \emph{literals} and the sub-expressions 
$C_i=\bigvee_{j=1}^{k_i}l_{ij}$ are called \emph{clauses}. 
The 3 in the 3-SAT specifies $k_i\le 3$.
The \emph{length}, $L$, of an instance is the number of 
AND/OR gates in the formula, namely $L=\sum_{i=1}^{m}k_i - 1$. We are interested in formulas $\phi$ with polynomial length.  For 3-SAT, we have $L = m_1-2m_{2}+3m_{3}-1$ where
$m_{2}$ is the number of 2-clauses and $m_{3}$ is the number of 3-clauses.

Given a 3-SAT formula $\phi$, the task is to determine whether there exists an 
assignment $\vec{x} \in \{0,1\}^n$ such that $\phi(x_1,\ldots, x_n)=1$,
in which case we call $\phi$ is satisfiable. For this task, one can assume without loss of generality that $m_1=0$, since otherwise all the singleton clauses imply definite values on the variables, and thus can be reduced.

It has been long known that 3-SAT is $NP$-complete, prohibiting it from having a
polynomial-time randomized algorithm based on the belief that $BPP \neq NP$. There were extensive attempts 
to find at least a sub-exponential algorithm for the problem without success, giving rise to the following conjecture.

\medskip
\noindent{\bf Exponential Time Hypothesis.} There is an $\epsilon>0$ such that the
time complexity of {\rm 3-SAT} is at least $(1+\epsilon)^{n} {\rm poly}(m)$.

\medskip

Due to subsequent improvements \cite{paturi1997satisfiability, paturi2005improved, schoning1999probabilistic, hertli20143} the current state-of-the-art 
3-SAT solver takes time $\approx 1.3^n$ (and in fact, the precise exponent is slightly worse). 
A classical algorithm breaking this bound would have a huge impact on theoretical computer science.

When trying to lower bound the running time of a 3-SAT solver in terms of $L$ and not of $n$, the ETH initially seems to be of little help,
as $L$ can be as large cubic in $n$. The following Sparsification Lemma, however, 
gives the desired $n$ to $L$ conversion.

\medskip
\begin{Lemma}[Sparsification Lemma, \cite{impagliazzo2001problems}]
    \label{lem:spa}% Sparsification lemma
    Given any $\epsilon > 0$, there is an algorithm $A_{\epsilon}$ that, on any {\rm 3-SAT} instance $\phi$ with $n$ variables, outputs a list $\ell = \phi_1,\cdots, \phi_k$
    of {\rm 3-SAT} instances in time $O_{\epsilon}(2^{\epsilon n})$, satisfying:
    \begin{itemize}
        \item $k\leq O_{\epsilon} (2^{\epsilon n})$;
        \item each formula $\phi_i$ has length at most $c(\epsilon)n$, where $c(\epsilon)$ does not depend on $n$;
        \item $\phi$ is satisfiable if and only if one of the generated sub-instances are satisfiable: $\phi = \bigvee_{i=1}^l \phi_i$.
    \end{itemize}
\end{Lemma}

\noindent Among the consequences to the Sparsification Lemma are Lemmas \ref{ethcons} and \ref{lem:main}. To prove Lemma \ref{lem:main} we 
must also compute the constants implicit in the Sparsification Lemma.  These proofs can be found in the Appendix.

\section{From 3-SAT to Clifford$+$T}\label{SATtoT}

\noindent In this section, we give a reduction from the 3-SAT problem to the problem of strongly simulating quantum circuits,
with the following properties.

\begin{itemize}
    \item[$(i)$]{For a $3$-SAT instance $\phi$ with length $L$, we construct a quantum circuit ${\cal C}'_{\phi}$} with at most $2L$ Toffoli gates and ${\rm poly}(L)$ NOT and CNOT 
    gates such that ${\cal C}'_{\phi}|x\rangle_I|0\rangle_A|0\rangle_B  =|x\rangle_I|0\rangle_A|\phi(x)\rangle_B$, where $A$ is the system of ancilla qubits and $B$ is a system of one single qubit.
    \item[$(ii)$]We choose a basis state (e.g. $|0\ldots 0\rangle$) which counts the number of assignments satisfying $\phi$ in its amplitude when running $\mathcal{C}_{\phi}:=H^{\otimes n}_I{\cal C}'_{\phi}H^{\otimes n}_I$, satisfying
    \[
    \langle 0 \ldots 0 | {\cal C}_{\phi} |0\ldots 0\rangle = { |\,  \{ x\in \{0,1\}^{n}\;  |\;  \phi(x) = 1\}\, |  \over 2^{n}}.
    \]
\end{itemize}

To achieve $(i)$, we introduce reversible computation, first investigated in~\cite{bennett1973logical}, for conversion of classical computation into unitary quantum circuits.

\noindent{\bf Reversible Circuits.} A reversible classical circuit consists of reversible gates. A reversible classical gate $F$ is simply an invertible function 
$F:\{0,1\}^d\rightarrow \{0,1\}^d$, for some $d$. Typically $d$ is one, two, or three. An important example of a reversible classical gate is the 
Toffoli gate $\toffoli(x,y,z)=(x,y,z\oplus (x\land y))$ acting on 3 bits. 

\noindent{\bf Universal Gate Set.} Throughout, our choice of universal gate set for reversible computation is
\[
{\cal G} = \{\toffoli,\cnot,\nnot\}.
\] 

A reversible circuit is readily a quantum circuit by replacing the gate set ${\cal G}$ using the quantum gate set $\{\toffoli, \cnot, X\}$. a reversible circuit satisfying property $(i)$ is defined as a \emph{tidy computation}.

\medskip

\noindent{\bf Tidy Computation.} We say a reversible circuit $C:\{0,1\}^{n+a(n)+1}\rightarrow \{0,1\}^{n+a(n)+1}$ \emph{tidily computes} a function $f:\{0,1\}^n\rightarrow \{0,1\}$ if 
\begin{equation}\label{reversible}
\forall \; x\in\{0,1\}^n, y\in\{0,1\}, \;\;\; C(x,0^{a(n)},y)=(x,0^{a(n)},y\oplus f(x)).
\end{equation}

\begin{Lemma}
    \label{lem:rev}% reversible circuit
    Suppose $\phi$ is a SAT formula with length $L$. Then, in time polynomial in $L$, we can construct a reversible circuit ${\cal C}'$ that tidily computes $\phi$ with
    at most $2L$ Toffoli gates.
\end{Lemma}

\medskip

Before proving Lemma~\ref{lem:rev}, we first define a specific form of reversible computation called diagonal computation.

\medskip
\noindent{\bf Diagonal Computation.} We say that a reversible circuit $C:\{0,1\}^n\rightarrow \{0,1\}^{n+a(n)+1}$ \emph{diagonally computes} a function $f:\{0,1\}^n\rightarrow\{0,1\} $ if for all $x$, $C(x,0,0)=(x,*,f(x))$.

\medskip

Diagonal computation is a special type of untidy computation, where the inputs on the input wires are preserved after the computation. This helps us compose diagonal computation circuits in a gate-efficient manner.

\begin{Lemma}
    \label{lem:block}
    Suppose $U_1$ and $U_2$ diagonally compute $f_1$ and $f_2$ with at most $a_1$ and $a_2$ ancilla wires and $t_1$ and $t_{2}$ Toffoli gates respectively, over the same set of input wires.
    Then there exists:
    \begin{enumerate}
        \item a circuit with at most $a_1+a_2+2$ ancilla wires and $t_1+t_2$ Toffoli gates that diagonally computes $f_{1} \wedge f_2$,
        \item a circuit with at most $a_1+a_2+2$ ancilla wires and $t_1+t_2$ Toffoli gates that diagonally computes $f_{1} \vee f_2$, and
        \item a circuit with at most $a_1$ ancilla wires and $t_1$ Toffoli that diagonally computes $\lnot f_1$.
        \item a circuit with at most $a_1+1$ ancilla wires and $2t_1$ Toffoli that tidily computes $f_1$.
    \end{enumerate}
\end{Lemma}

\begin{proof}[Proof of Lemma~\ref{lem:block}]
    The reversible diagonal AND, OR and NOT of the circuits, with the appropriate sizes and widths, are constructed explicitly below. Moreover, a circuit diagonally computing $f_1$ is also untidily computing $f_1$, thus 4.\ can be proven using conversion from untidy circuits to tidy circuits.
    \begin{figure}[H]
        \centering
        \begin{subfigure}[h!]{0.29\textwidth}
        \includegraphics[width = \textwidth]{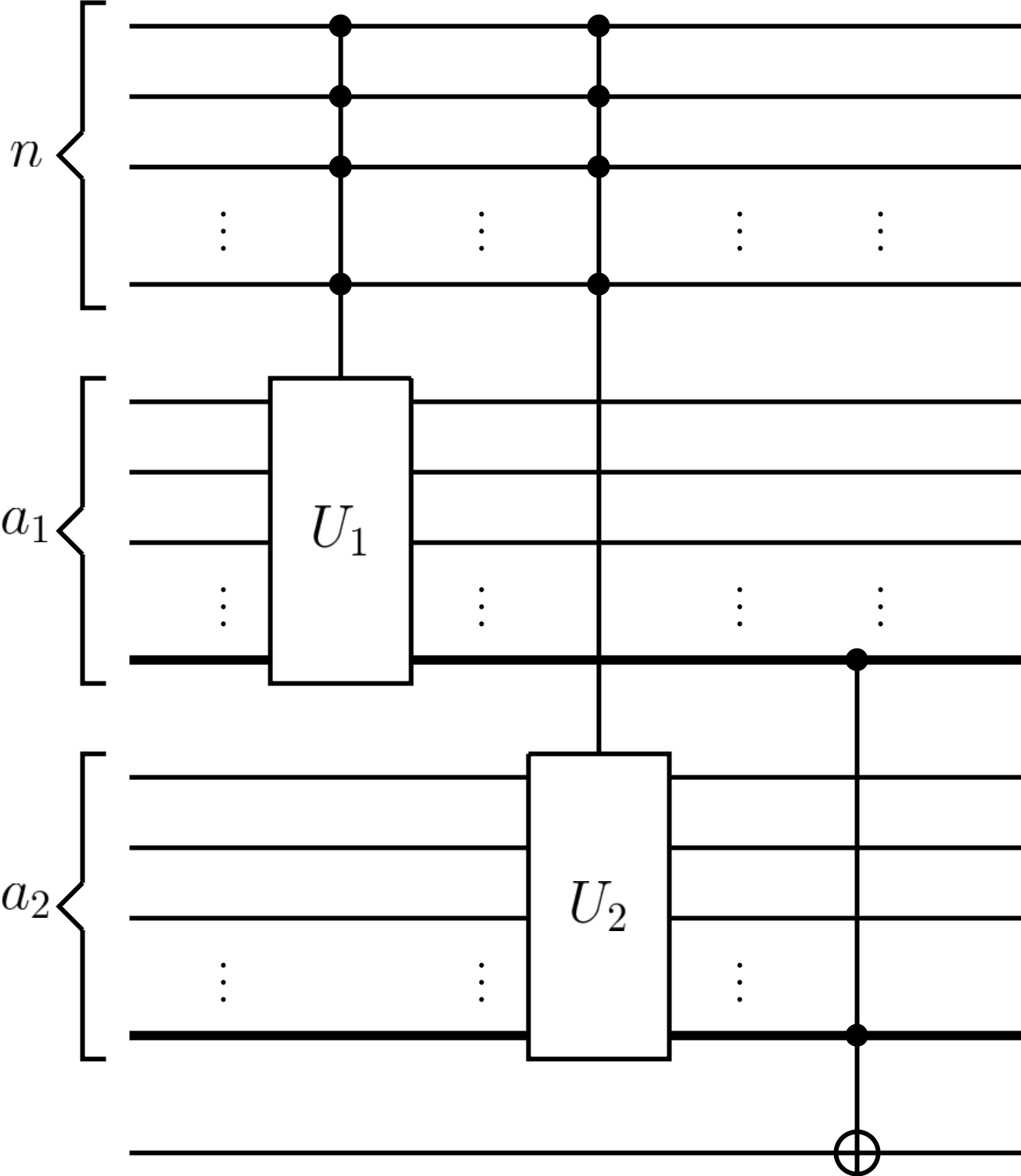}
        \caption{Circuit for $f_1\wedge f_2.$}
        \end{subfigure}
        ~
        \begin{subfigure}[h!]{0.36\textwidth}
        \includegraphics[width = \textwidth]{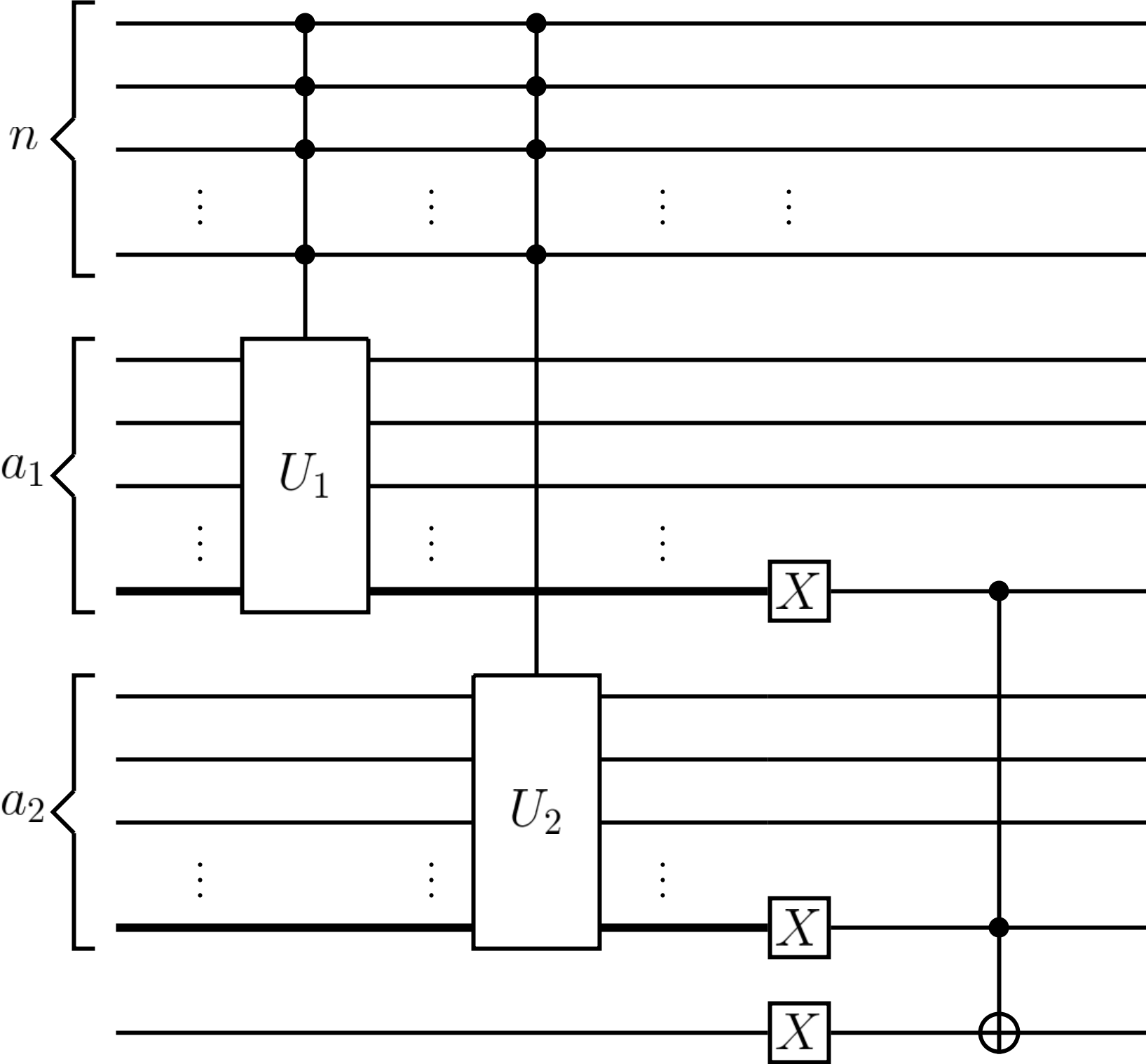}
        \caption{Circuit for $f_1\vee f_2$.}
        \end{subfigure}
        ~
        \centering
        \begin{subfigure}[h!]{0.3\textwidth}
        \centering
        \includegraphics[width = \textwidth]{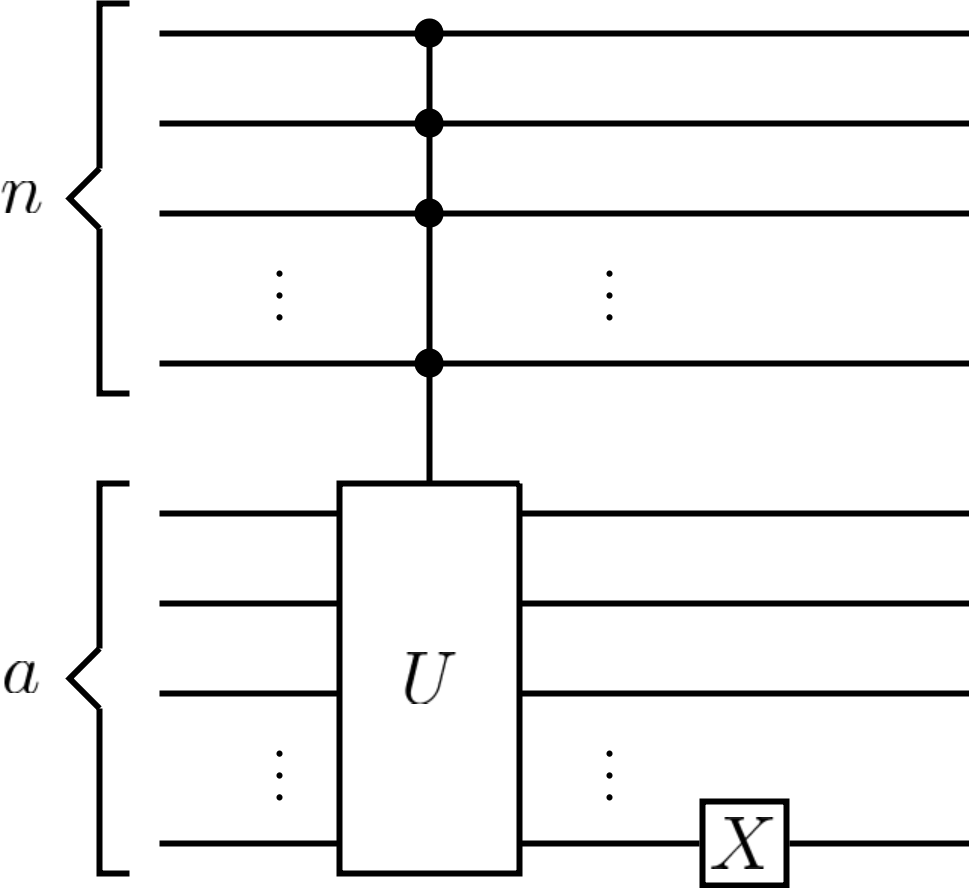}
        \caption{Circuit for $\lnot f_1$.}
        \end{subfigure}
    \end{figure}
\end{proof}

Property $(i)$ is then a direct application of Lemma~\ref{lem:rev}. For property $(ii)$, we take the following from~\cite{nielsen2002quantum}:

\begin{Lemma}
    \label{lem:toftot}% toffoti to T
    A Toffoli gate can be written as composition of $7$ $T$-gates and $8$ Clifford gates.
\end{Lemma}

\begin{proof}[Proof of Lemma~\ref{lem:toftot}]
    The circuit cmoputing a Toffoli gate with $7$ $T$-gates and $8$ Clifford gates is explicitly constructed below in Fig.~\ref{fig:toffoli_circuit}.
    \begin{figure}[H]
        \centering
        \includegraphics[width=0.7\textwidth]{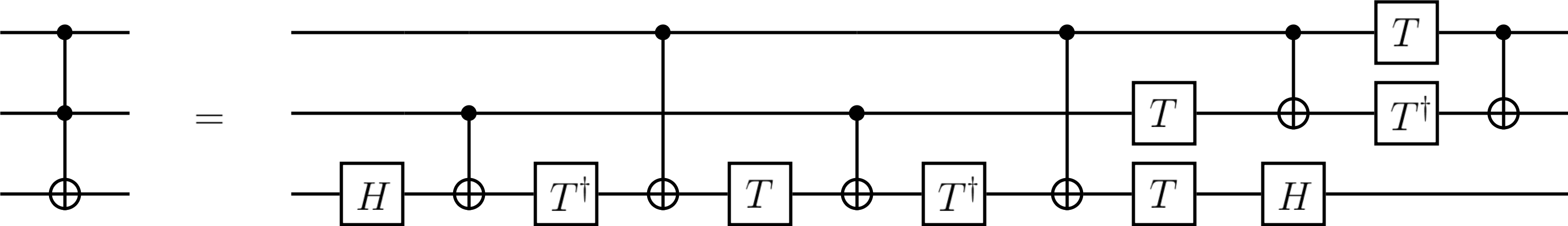}
        \caption{Explicit construction of Toffoli gate with Clifford+$T$ gates with minimum $T$-count.}
        \label{fig:toffoli_circuit}
    \end{figure}
\end{proof}

Combining Lemmas~\ref{lem:rev} and~\ref{lem:toftot}, we obtain:

\begin{Corollary}
    \label{col:sstosat}% strong simulation to sat
    Given a {\rm 3-SAT} formula $\phi$ with $n$ variables and length $L$, one can efficiently construct a quantum circuit ${\cal C}_{\phi}$ consisting only of Clifford gates and at most $14L$ $T$-gates, so that
    $$\langle 0|C_{\phi}|0\rangle = \frac{|\{x\in \{0,1\}^n|\phi(x)=1\}|}{2^n}.$$
    Consequently, strong simulation of such a circuit up to accuracy $2^{-n}/2$ is as hard as determining whether $\phi$ is satisfiable.
\end{Corollary}

\begin{proof}[Proof of Theorem \ref{theorem:main}] Assume that Theorem \ref{theorem:main} is false. 
Then for any $\epsilon>0$ we have a quantum circuit simulator that runs in time $O(2^{{\epsilon\over 14} N})$, where $N$ is the number of $T$ gates and the size of 
the circuit is $O(N)$.
If we apply the reduction in Corollary~\ref{col:sstosat}, we can solve the 3-SAT problem in time $O(2^{{\epsilon}L})$, contradicting 
Lemma \ref{ethcons}. \end{proof}

\begin{proof}[Proof of Theorem \ref{thm:main1}] According to Corollary~\ref{col:sstosat} a 3-SAT instance of length $L$ reduces to the strong simulation problem of a linear size
quantum circuit with Clifford gates and at most $N=14L$ $T$-gates.  If there were a solution to the simulation problem with running time $O(2^{2.2451\times 10^{-8}N})$, then
any 3-SAT instance of length $L$ could be solved in time $O(2^{2.2451\times 10^{-8}\times 14L})= O(2^{3.1432\times 10^{-7}L} )$,
contradicting Lemma \ref{lem:main}. \end{proof}

%% file: appendix.tex
\begin{appendices}

\section{Proof of the Sparsification Lemma}

\medskip
Given a 3-SAT instance $\phi = C_1\land C_2\land\dots C_m$, we identify each clause $C_1,C_2,\dots C_m$ as a 
subset of all literals $\{x_1,\lnot x_1,x_2,\lnot x_2,\dots, x_n,\lnot x_n\}$. We start from a simple Boolean identity:
$$(a\lor b)\land (a\lor c) = a\lor (b\land c).$$
This identity implies the following Lemma. 

\begin{Lemma}\label{sunlemma}
For an arbitrary subset $\{C_1,\dots, C_{m'}\}$ of clauses of $\phi$
 and for $C := \bigcap_{i=1}^{m'} C_i$, we have $\; \phi = \phi_1\lor \phi_2$, where
\begin{align*}
       \phi_1 & = C\land C_{m'+1}\land C_{m'+2}\land\dots\land C_m,\\
    \phi_2 & = (C_1\setminus C)\land (C_2\setminus C)\land\dots\land (C_{m'}\setminus C)\land C_{m'+1}\land C_{m'+2}\land\dots\land C_m.
\end{align*}
\end{Lemma}

Given that $\phi$ is a 3-SAT instance, both $\phi_{1}$ and $\phi_{2}$ are also 3-SAT instances. Moreover, we also have the following.

\begin{Lemma}\label{lengthlemma} Let $\phi$, $\phi_{1}$, $\phi_{2}$ be defined as in Lemma \ref{sunlemma}. Then neither of the new instances has length greater than the original:
$L(\phi) \ge L(\phi_{1}), \; L(\phi_{2})$.
\end{Lemma}

\noindent{\bf Sunflowers.}  We call a collection $C_1,\dots, C_{m'}$ of clauses a $(k,h)$-\emph{sunflower} (with $h>0$) if
\begin{itemize}
    \item Each $C_i$ contains exactly $k$ literals, and
    \item $C:= \bigcap_{i=1}^{m'} C_i$ contains $h$ literals. 
\end{itemize}
$C$ is then called the \emph{heart} of the sunflower and the collection $\{C_1\setminus C, \cdots, C_{m'}\setminus C\}$ of clauses are called \emph{petals}. The algorithm for sparsification then 
keeps a set of current 3-SAT formulas whose disjunction is $\phi$. Moreover, it repeatedly replaces a formula in this set
with two formulas as long as it finds a collection of its clauses that is a large sunflower. One of these new formulas 
is obtained from the original by replacing the sunflower with its petals, while the other is obtained by replacing the sunflower with its heart. 

\medskip
\noindent{\textsc{Sparsification Algorithm}}. For 3-SAT instances, there are three kinds of sunflowers: $(2,1)$-sunflowers, $(3,2)$-sunflowers, and $(3,1)$-sunflowers. Consider the following algorithm parametrized by $\theta_1, \theta_2$: $1\le \theta_1\le \theta_2$ to be determined. Call a sunflower \emph{good} if it is a $(2,1)$- or $(3,2)$-sunflower of size at least $\theta_1$, or a $(3,1)$-sunflower of size at least $\theta_2$. Among good sunflowers, $(2,1)$-sunflowers have higher priority than $(3,2)$-sunflowers, and $(3,2)$-sunflowers have higher priority than $(3,1)$-sunflowers.
Our sparsification algorithm first creates an empty list $\ell$, which is a global variable, and then calls (once) the recursive \textsc{Sparsify} algorithm below.
\begin{algorithm}[H]
    \caption{The Algorithm \textsc{Sparsify}.}
    \label{alg:spar}
    \begin{algorithmic}[1]
    \Procedure{Sparsify}{$\phi$}
        \If {$\phi$ does not contain a good sunflower}
            \State append $\phi$ to $\ell$.
        \Else
            \State let $C_1,C_2,\dots, C_{m'}$ be a good sunflower in $\phi$ with the highest priority and let $C$ be the heart
            \State $\phi_h = $\textsc{Reduce}($C\land C_{m'+1}\land C_{m'+2}\land\dots\land C_m$)
            \State $\phi_p =$\textsc{Reduce}($(C_1\setminus C)\land (C_2\setminus C)\land\dots\land (C_{m'}\setminus C)\land C_{m'+1}\land C_{m'+2}\land\dots\land C_m$)
            \State  \textsc{Sparsify}($\phi_h$); \textsc{Sparsify}($\phi_p$)
        \EndIf
    \EndProcedure
    \end{algorithmic}
\end{algorithm}
\begin{algorithm}[H]
    \caption{The Algorithm \textsc{Reduce}.}
    \label{alg:red}
    \begin{algorithmic}[2]
    \Procedure{Reduce}{$\phi$}
        \While {$\phi$ contains two clauses $C_i$ and $C_j$, $C_i\subseteq C_j$}
            \State remove $C_j$ from $\phi$
        \EndWhile
        \State return $\phi$
    \EndProcedure
    \end{algorithmic}
\end{algorithm}
Note that the algorithm traverses through a binary recursion tree rooted at $\phi$, where each node corresponds to a 3-SAT formula. 
The set of 3-SAT formulae corresponding to leaf nodes is the collection of instances in $\ell$, which is the list we needed to construct. 
A recursive application of Lemma \ref{sunlemma} gives that $\bigvee_{\phi_{i}\in \ell} \phi_{i} = \phi$. We further prove the following.
\begin{itemize}
    \item Each leaf node corresponds to a formula of length at most $\eta(\theta_1,\theta_2)n$, where 
    \[
    \eta(\theta_1,\theta_2):= 2(\theta_1 + \theta_2);
    \]
    \item There are at most $2^{\gamma(\theta_1,\theta_2)n}$ nodes in the 
    tree, where 
    \[
    \gamma(\theta_1,\theta_2) := 4\theta_1\times H(\frac{1}{4\theta_1^2}+\frac{1}{\theta_2}); \;\; \;\;\; H(p):=-p\log_2 p-(1-p)\log (1-p).
    \]
    \item Our algorithm runs in time $O(2^{\gamma(\theta_1,\theta_2)}poly(n))$. This follows immediately from the above.
\end{itemize}

\medskip
\noindent{\bf Maximum Length of Each Leaf Node.} Given a $3$-SAT instance $\phi^{\ast}$, denote the number of $2$-clauses by $m_2$ 
and the number of $3$-clauses by $m_3$. Clearly, $L(\phi^{\ast}) = 2m_2(\phi^{\ast})+3m_3(\phi^{\ast})-1$. Let $d_j(\phi^{\ast})$ be the maximum number of clauses of size $j$ with an nonempty intersection. 
We have the following observation (by counting the total number of literals in 2-, respectively 3-, clauses):
\begin{eqnarray*}
 2 m_2(\phi^*)\leq 2n\cdot d_2(\phi^*) \\
 3  m_3(\phi^*)\leq 2n\cdot d_3(\phi^*).
 \end{eqnarray*}

For a formula $\phi^*$ on a leaf node, since there are no $(2,1)$-sunflowers of size at least $\theta_1$, we have $d_2(\phi^*)<\theta_1$. 
Similarly, $d_3(\phi^*)<\max(\theta_{1},\theta_{2}) = \theta_{2}$. Together with Lemma \ref{lengthlemma}, these give:
\begin{equation}\label{leaflength}
L(\phi^*) \; = \; 2m_2(\phi^*)+3m_3(\phi^*)-1 \; < \;  2(\theta_1 +\theta_2) m(\phi^*) \; \le \; 2(\theta_1 +\theta_2) m(\phi).
\end{equation}

\medskip
\noindent{\bf Number of Leaf Nodes.} To upper bound the number of leaf nodes, we need the notion of \emph{immigrant clauses}. For any formula on some node of the recursion tree, call a clause \emph{immigrant} if that clause is not present in the root. For any path from the root to a leaf, all immigrant clauses that newly appear are distinct (i.e. it cannot happen
that an immigrant clause disappears and then reappears later). This leads to the following observation.

\begin{Observation}
    \textsc{Reduce} only happens when a newly introduced immigrant clause is contained in previously present clauses.
\end{Observation}

The high-level idea of the proof is as follows: we show that there are at most a linear number of immigrant clauses ever introduced. Since in each round at least one immigrant 
clause is introduced, the recursion tree has linear depth. Moreover, many immigrant clauses are created whenever the petals of a sunflower are taken, and so there must be few such steps, 
further restricting the number of leaves.

Let $r_2(\phi^*)$ be the maximum number of immigrant $2$-clauses with nonempty intersection. Clearly $r_2(\phi^*)\leq d_2(\phi^*)$. 
The following holds for every node in the recursion tree.

\begin{Lemma}
    For every formula $\phi^*$ in the recursion tree, $r_2(\phi^*)\leq 2\theta_{1}-1$.
\end{Lemma}

\begin{proof}
    The proof follows by induction from top to bottom. For the root, $\phi$, 
we have $r_2(\phi) = 0$. 
Next consider a non-top node $v$ on which a new 
immigrant $2$-clause is created, and the corresponding formula $\phi^*$. There are two cases to consider.
    \begin{itemize}
        \item $\phi^*$ takes the heart of a $(3,2)$-sunflower from its parent $\phi'$. Since $\phi'$ does not have a $(2,1)$-sunflower of size $\theta_1$, $d_2(\phi')\le \theta_1 - 1$. 
        By only adding one new $2$-clause, we have that
        \[
        r_2(\phi^*)\leq d_2(\phi^*)\leq d_2(\phi')+1 \le \theta_1 \le 2\theta_{1}-1.
        \]
        \item $\phi^*$ takes the petals of a $(3,1)$-sunflower from its parent $\phi'$. Similar to the former case, $d_2(\phi')\le \theta_1 - 1$. Assume that $r_2(\phi^*)\geq 2\theta_1$. Then there exists a literal $y$ which appeared in at least $\theta_1+1$ of the newly-formed petals. However, this is not possible as there would be a $(3,2)$-sunflower of size at least $\theta_1+1$ in $\phi'$, and the algorithm would choose that sunflower instead of a $(3,1)$-sunflower.
    \end{itemize}
\end{proof}

This leads to the following observation.

\begin{Observation}
    An immigrant clause of size one can only reduce at most $2\theta_{1}-1$ immigrant clauses of size two.
\end{Observation}

There are at most $n$ immigrant $1$-clauses introduced 
(literals corresponding to a single variable can be immigrant at most once), and so there are at most $(2\theta_1-1) n$ immigrant $2$-clauses reduced by them (because in each reduction at most 
$2\theta_{1}-1$ immigrant clauses are eliminated). 
When the algorithm arrives at a leaf, $\phi^*$, because $r_2(\phi^*)\leq 2\theta_{1}-1$, we have that the total number of immigrant $2$-clauses that remains is at most 
${(2\theta_{1}-1) 2n\over 2} < (2\theta_{1}-1)n$. Thus the number of immigrant $2$-clauses ever introduced is at most $(4\theta_1-2) n$. In each step going down in the recursion tree
at least one new immigrant one- or two- clause was created, and so depth of the recursion tree is at most $(4\theta_1-2) n + n < 4\theta_1 n$.
This alone would not be sufficient to get a good estimate on the number number of leaves, but we further observe the following.

Each time the petals of a sunflower are taken, either at least $\theta_1$ $1$-clauses are introduced, or at least $\theta_2$ $2$-clauses are introduced. 
Therefore, the number of petals taken along a path
from the root to a leaf is at most $\frac{n}{\theta_1}+\frac{4\theta_1n}{\theta_2}$. This gives the bound
\begin{equation}\label{leafeq}
\sum_{i=0}^{\frac{n}{\theta_1}+\frac{4\theta_1n}{\theta_2}}\binom{4\theta_1n}{i}\leq 2^{\gamma(\theta_1,\theta_2)n},
\end{equation}
on the number of leafs,
where $\gamma(\theta_1,\theta_2)\leq 4\theta_1H(\frac{1}{4\theta_1^2}+\frac{1}{\theta_2}).$

\medskip

\noindent {\bf Optimization.} First, note that Equation \ref{leafeq} and $H(p) /  p(1+ \log_2 {1\over p}) \longrightarrow 1$ at $p=0$
given that $\gamma(\theta_1,\theta_2)$ can be arbitrarily small at $\theta_2 = 4 \theta_{1}^2$ and for $\theta_{1}$ sufficiently large.
This, together with Equation \ref{leaflength}, gives the Sparsification Lemma.

\medskip

Next, we compute the values for $\theta_{1}$ and $\theta_{2}$ that optimize the hardness reductions from instances whose size parameters are expressed in terms of $n$ 
to instances whose size parameters are expressed in terms of $L$.

\section{Proof of Lemma~\ref{ethcons}}

\begin{proof}
    Suppose that for each $a>0$, contradictory to Lemma~\ref{ethcons}, there is an algorithm $\textsc{Solve}_a$ that solves 3-SAT in time $O(2^{aL})$. 
    We show that the existence of such a family of algorithms implies 
    the existence of a family of algorithms that solve $3$-$SAT$ in time $O(2^{\epsilon n})$ for every $\epsilon > 0$, 
    contradicting the ETH.
    Given $\epsilon > 0$, set $\epsilon' = \epsilon/2$. Consider the following algorithm.
    \begin{enumerate}
        \item Given a $3$-$SAT$ instance $\phi$ over $n$ variables, run the 
        \textsc{Sparsification Algorithm} and get $k$ $3$-$SAT$ instances $\phi_1,\dots, \phi_k$, where $k\leq 2^{\epsilon' n}$, each of length at most $c(\epsilon')n$.
        \item Solve every instance $\phi_1,\dots, \phi_k$ using the algorithm $\textsc{Solve}_{\epsilon'/c(\epsilon')n}$.
        \item If any of the $\phi_i$ are satisfiable, output $1$; otherwise output $0$.
    \end{enumerate}
    By the Sparsification Lemma, Step 1 takes time $2^{\epsilon' n}poly(n)$ time. Solving 
    an instance in Step 2 takes time $2^{{\epsilon' \over c(\epsilon')} c(\epsilon')n }$ since $L \le c(\epsilon')n$. 
    The total running time of Step 2 is then $2^{(\epsilon' + \epsilon')n} = 2^{\epsilon n}$.
    Finally, Step 3 combines the results from Step 2, and so takes time $O(2^{\epsilon' n})$. The overall running time is then
    dominated by $2^{\epsilon n}$.
\end{proof}

\section{Proof of Lemma~\ref{lem:main}}
The proof of Lemma~\ref{lem:main} is very similar to the proof of Lemma~\ref{ethcons}, except that we now have to calculate the explicit constants.
 Assume we had a $3$-SAT solver \textsc{Solve} that runs in time $o(2^{3.1432\times 10^{-7}L})$ 
 on instances of length $L$. From it, we construct a $3$-SAT 
 solver  \textsc{Solve}' that runs in time $o(1.3^n)$, beating the current best $3$-SAT solver.
Let  $\theta_1=109.395$ and $\theta_2=58367.2$ . \textsc{Solve}' will then
    \begin{itemize}
        \item Run the \textsc{Sparsification Algorithm} on input $\phi$ to get a list $\ell$ of $2^{\gamma(\theta_1,\theta_2) n}$ sparse instances
        in time $O(2^{\gamma(\theta_1,\theta_2)n}poly(n))$, each with length $\eta(\theta_1,\theta_2)n$. We have $\bigvee_{\phi_{i}\in \ell} \phi_{i} = \phi$.
        \item Use  \textsc{Solve} to solve each instance $\phi_{i}$ in time $o(2^{3.1432\times 10^{-7}\eta(\theta_1,\theta_2)n})$. The total running time is then less than
        \[
         2^{3.1432\times 10^{-7}(\eta(\theta_1,\theta_2)+\gamma(\theta_1,\theta_2))n}.
        \]
        \item If any of the instances are satisfiable, it outputs $1$, otherwise it outputs $0$. This step takes time $O(2^{\gamma(\theta_1,\theta_2)n})$.
    \end{itemize}
    The dominating term in the running time is then $2^{3.1432\times 10^{-7}(\eta(\theta_1,\theta_2)+\gamma(\theta_1,\theta_2))n}$. An easy calculation shows that
     \textsc{Solve}' runs in $o(1.3^n)$, beating the current best bound.

\end{appendices}